\documentclass[11pt, a4paper]{article}

\usepackage{natbib}
 \bibpunct[, ]{[}{]}{,}{n}{}{,}%

\usepackage{typearea}
\paperwidth 8.5in \paperheight 11in \typearea{16}

\usepackage{times}
\usepackage{theorem,latexsym,graphicx}
\usepackage{amsmath,amssymb,enumerate}
\usepackage{xspace}
\usepackage{bm}
\usepackage{ifpdf}
\usepackage{color}
\usepackage{algorithm, algorithmic}
\usepackage{paralist}

\newcommand{\ignore}[1]{}

\definecolor{Darkblue}{rgb}{0,0,0.4}
\definecolor{Brown}{cmyk}{0,0.81,1.,0.60}
\definecolor{Purple}{cmyk}{0.45,0.86,0,0}

\newcommand{\mydriver}{hypertex} \ifpdf
 \renewcommand{\mydriver}{pdftex}
\fi
\usepackage[breaklinks,\mydriver]{hyperref}
\hypersetup{colorlinks=true,
            citebordercolor={.6 .6 .6},linkbordercolor={.6 .6 .6},%
citecolor=Darkblue,urlcolor=black,linkcolor=Darkblue}
		
\newcommand{\lref}[2][]{\hyperref[#2]{#1~\ref*{#2}}}


\newtheorem{theorem}{Theorem}[section]
\newtheorem{definition}[theorem]{Definition}
\newtheorem{proposition}[theorem]{Proposition}
\newtheorem{assumption}[theorem]{Assumption}

\newtheorem{lemma}[theorem]{Lemma}

\newtheorem{claim}[theorem]{Claim}

\newtheorem{corollary}[theorem]{Corollary}

\allowdisplaybreaks

\newenvironment{proof}{

\noindent{\bf Proof:}} {\hfill$\blacksquare$

}

\newcommand{\R}{\mathbb{R}}

\def\cI{{\cal I}}
\def\cD{{\cal D}}

\def\cN{{\cal N}}
\def\lpo{\mathsf{LP^*}}
\def\oH{\overline{H}}
\def\sse{\subseteq}

\begin{document}

\title{Approximation Algorithms for Inventory Problems with \\Submodular or Routing Costs}

\author{ Viswanath Nagarajan\thanks{Industrial \& Operations Engineering, University of Michigan, Ann Arbor, MI, \{viswa, shicong\}@umich.edu}
\and Cong Shi$^\ast$ }
\date{}

\maketitle

\begin{abstract}
We consider the following two deterministic inventory optimization problems over a finite planning horizon $T$ with  non-stationary demands.
\begin{itemize}
\item {\em Submodular Joint Replenishment Problem.} This involves multiple item types and a single retailer who faces demands. In each time step, any subset of item-types can be ordered incurring a joint ordering cost which is submodular. Moreover, items can be held in inventory while incurring a holding cost. The objective is find a sequence of orders that satisfies all demands and minimizes the total  ordering and holding costs. 

\item {\em Inventory Routing Problem.} This involves a single depot that stocks items, and multiple retailer locations facing demands. 
In each time step, any subset of locations can be visited using a vehicle originating from the depot. There is also cost incurred for holding items at any retailer. The objective here is to satisfy all demands while minimizing the sum of routing and holding costs.
\end{itemize}

We present a unified approach that yields $\mathcal{O}\left(\frac{\log T}{\log\log T}\right)$-factor approximation algorithms for both problems when the holding costs are polynomial functions. A special case is the classic linear holding cost model, wherein this is the first sub-logarithmic approximation ratio for either problem.  

\end{abstract}

\section{Introduction}
\label{sec:intro}

Deterministic inventory theory provides streamlined optimization models that attempt to capture tradeoffs in managing the flow of goods through a supply chain. We consider two classical models in deterministic inventory theory: the \emph{Joint Replenishment Problem} (JRP) and the \emph{Inventory Routing Problem} (IRP). These inventory models have been studied extensively in the literature (see, e.g., \cite{aksoy1993multi}, \cite{joneja1990joint}) and  recently there has been significant progress on many variants of these models (see, e.g., \cite{levi2006primal}, \cite{levi2008constant}, \cite{nonner2013efficient}, \cite{CELS2014}, \cite{FNR14}). In this paper, we present a unified approach that yields approximation algorithms for both models with generalized cost structure -- the JRP with submodular setup cost and the IRP with arbitrary embedding metric.  

The JRP with deterministic and non-stationary demand is a fundamental yet notoriously difficult problem in inventory management. In these models, there are multiple item types, and we need to coordinate a sequence of (joint) orders to satisfy the demands for different item types before their respective due dates. Ordering inventory in a time period results in setup costs (or fixed ordering costs), and holding inventory before it is due results in holding costs. The objective is to find a feasible ordering policy to satisfy every demand point on time over a finite planning horizon so as to minimize the sum of setup and holding costs. The JRP is a natural extension of the classical economic lot-sizing model that considers the optimal trade-off between setup costs and holding costs for a single item type (see \cite{wagner1958dynamic}). With multiple item types, the JRP adds the possibility of saving costs via coordinated replenishment, a common phenomenon in supply chain management.

Most of the literature on deterministic JRP is on the \emph{additive} joint setup cost structure. Under this structure, there is a one-time setup cost if any item type is ordered, and there is an individual item setup cost for each item type ordered; the joint setup cost for this particular order is simply the sum of the one-time setup cost and these individual item setup costs. The additive joint setup cost structure loses significant amount of modeling power and flexiblity (see \cite{queyranne1986polynomial}, \cite{federgruen1992joint}, \cite{teo2001multistage}, \cite{CELS2014}). In this paper, we adopt the joint setup cost structure introduced recently in \cite{CELS2014} that satisfies two natural properties known as \emph{monotonicity} and \emph{submodularity}. The monotonity property means that the joint setup cost increases with the set of item types ordered. The submodularity property captures economies of scale in ordering more item types, i.e., the marginal cost of adding any specific item type to a given order decreases in the set of item types included.

The IRP is also a classical problem in inventory management that captures the trade-off between the holding costs for inventory and the routing costs for replenishing the inventory at various locations in a supply chain (see, e.g., \cite{Burns1985}, \cite{Coelho2014}, \cite{FZ1984}, \cite{FNR14}). The problem involves multiple item types that are stocked in a single depot, that must be shipped to meet the demand for these item types arising at multiple retailers specified over the course of a planning horizon. Similar to the JRP, the costs of holding a unit of inventory at each retailer are specified to compute the inventory holding costs. Different than the JRP, we consider transportation (or vehicle routing) costs in some metric defined by the depot and retailers in the IRP, instead of joint setup costs considered in the JRP.

\subsection{Main Results and Contributions}
We present a unified  approach that yields  $\mathcal{O}\left(\frac{\log T}{\log\log T}\right)$-approximation algorithms for both the JRP with submodular setup costs and the IRP with any embedding metric, when the holding costs are polynomial functions (which subsumes conventional linear costs as a special case). This is the first sub-logarithmic approximation ratio for either problem under these cost structures. 

We remark that if the setup cost function in submodular-JRP is time-dependent  then the problem (even with zero holding costs) becomes as hard to approximate as set cover~\cite{F98}. The same observation is true if the metric in IRP is time-dependent. So our sub-logarithmic ratio approximation algorithm relies crucially on the uniformity of these costs over time.

For the submodular JRP, \citet{CELS2014} obtained constant-factor approximation algorithms under several special submodular functions (i.e., tree, laminar and cardinality). In contrast, we consider general submodular functions with special (polynomial) holding costs.  

For the IRP, \citet{FNR14} considered a restricted class of ``periodic policies'' and obtained a constant-factor approximation algorithm. Whereas our result is for arbitrary policies and polynomial holding costs. 

A straightforward modification of our algorithm for polynomial holding costs also yields $\mathcal{O}(\log T)$-approximation algorithms for submodular JRP and IRP with arbitrary (monotone) holding costs. The submodular JRP result improves upon the approximation ratio of $\mathcal{O}(\log(NT))$ by \cite{CELS2014}. The IRP result is incomparable to the $\mathcal{O}(\log n)$ approximation ratio mentioned in~\cite{FNR14}, where $n$ is the number of retailers. 


\subsection{Our Approach}
At a high-level, the algorithm for submodular JRP has the following steps. (The algorithm for IRP is very similar -- we in fact present an algorithm for a unified problem formulation.) First, we  solve a natural time-indexed LP relaxation that was also used in~\cite{CELS2014}. Then we construct a ``shadow interval'' for each demand point that corresponds to fractionally ordering half unit of the item. We also stretch each shadow interval appropriately (depending on the degree of the holding cost function)
so as to obtain an optimal trade-off between holding and setup costs: this is what results in the $\mathcal{O}\left(\frac{\log T}{\log\log T}\right)$ approximation ratio. Next, we partition these stretched intervals into multiple groups based on 
well-separated widths. Finally we place a separate sequence of orders for each group, and argue using submodularity of the setup cost function that the total setup cost of each group is bounded by the LP setup cost. This step relies on the {\em fractional subadditivity} property of submodular functions. 

It turns out that we do not require the full strength of submodular functions: the algorithm and analysis work even for functions satisfying an approximate notion of fractional subadditivity (see Definition \ref{def_afs}) as long as the natural LP relaxation can be solved approximately. This allows us to also obtain an approximation algorithm for IRP since the TSP cost function satisfies $1.5$-approximate  fractional subadditivity and there is a $2+o(1)$ approximation algorithm for its LP relaxation (see Section~\ref{sec_solve} for details).

We believe that some of our techniques may be useful in obtaining a constant factor approximation algorithm for both problems in their full generality.

\subsection{Literature review} \label{literature}
As mentioned earlier, most of the existing literature on deterministic JRP with non-stationary demand uses the additive joint setup cost structure. \citet{arkin1989computational} showed that the additive JRP is NP-hard. \citet{nonner2013efficient} further showed that the additive JRP  is in fact APX-hard with nonlinear holding cost structure. There have been several approximation algorithms for the additive JRP (see \cite{CELS2014} and the references therein). The state-of-the-art approximation algorithms for the additive JRP are due to \cite{levi2006primal}, \cite{levi2008constant} and \cite{bienkowski2013better}, with approximation ratios 2, 1.80 and 1.791, respectively.

Due to the limited modeling power of the additive JRP, \citet{CELS2014} first studied the submodular JRP in which the setup costs are submodular. They gave an $\mathcal{O}(\log(NT))$-approximation algorithm for the general submodular JRP (where $N$ is number of items and $T$ is number of periods). They also analyzed three special cases of submodular functions which are laminar, tree and cardinality cases. They showed that the laminar case can be solved optimally in polynomial time using dynamic programming, and obtained a 3-approximation for the tree case and a 5-approximation for the cardinality case. Our work contributes to the literature by giving approximation algorithms for the general submodular JRP with special holding cost structures. 

The IRP has also been studied extensively in the literature (see \cite{Burns1985}, \cite{Coelho2014}, \cite{FZ1984}, \cite{FNR14} for an overview of this problem). The problem can be cast a mathematical program (see, e.g., \cite{Campbell2004}) and solution approaches typically involve heuristics that trade-off between holding and transportation costs (see \cite{Anily1990,Anily1993}, \cite{Chan1998}, \cite{Chien1989}, \cite{VM1997}, \cite{Bertazzi2008}). 
Closer to our work, \citet{FNR14} gave constant factor approximation algorithms for the IRP restricting to periodic schedules. In contrast, our results do not require the schedule to be periodic but require polynomial holding costs.

\subsection{Structure of this paper and some notations}
We organize the remainder of the paper as follows. In Section~\ref{sec_model}, we present a unified formulation for the submodular JRP and the IRP with arbitrary embedding metric, and state our main result. In Section~\ref{sec_algorithm}, we propose a unified approximation algorithm for both problems. In Section~\ref{sec_solve}, we discuss how to solve the LP relaxation efficiently. We conclude our paper in Section~\ref{sec_cf}.

Throughout the paper, we use the notation $\lfloor x \rfloor$ and $\lceil x \rceil$ frequently, where $\lfloor x \rfloor$ is defined as the largest integer value which is smaller than or equal to $x$; and $\lceil x \rceil$ is defined as the smallest integer value which is greater than or equal to $x$. Additionally, for any real numbers $x$ and $y$, we denote $x^+=\max\{x, 0\}$, $x \vee y=\max\{x,y\}$, and $x \wedge y=\min\{x, y\}$. The notation $:=$ reads ``is defined as". 

\section{A Unified Formulation for the JRP and the IRP} \label{sec_model}
In this section, we formally describe a unified problem statement that includes two classical deterministic inventory problems as special cases, i.e., the joint replenishment problem (JRP) with submodular setup costs and the inventory routing problem (IRP) with arbitrary embedding metric. We also present a unified framework for this problem, and state our main result.

\subsection{Problem Statement} 
There are $N$ elements (e.g., item types in the JRP or retailers in the IRP) that are needed to serve external demands over a finite planning horizon of $T$ periods; these elements are denoted by the ground set $\mathcal{N}=\{1,\ldots, N\}$, and the time periods are denoted by the set $\mathcal{T}=\{1,\ldots, T\}$. For each time period $t\in \mathcal{T}$ and each element $i \in \mathcal{N}$ , there is a known demand $d_{it} \ge 0$ units of that element. We use $\mathcal{D}$ to denote the set of all strictly positive demand points $(i,t)$ with $d_{it} > 0$. To satisfy these demands, an order may be placed in each time period. Each demand point $(i,t) \in \mathcal{D}$ has to be served by an order containing element $i$ before or at time period $t$, i.e., no backlogging or lost-sales are allowed. 

The inventory system incurs two types of cost -- the joint ordering cost and the holding cost. 
\begin{itemize}
\item The joint ordering cost is a function of the elements that place strictly positive orders in any given period. More specifically, for any time period $t$ and a subset of elements $S \subseteq \mathcal{N}$, the joint ordering cost of ordering demand for elements in $S$ in period $t$ is a function of $S$, which is denoted by $f(S)$. 
\item Because the setup cost of ordering an element is independent of the number of units ordered, there is an incentive to place large orders to meet the demand not just for the current time period, but for subsequent time periods as well. This is balanced by a cost incurred by holding inventory over time periods. We use $h_{st}^{i}$ to denote the holding cost incurred by ordering one unit of inventory in period $s$, and using it to meet the demand for element $i$ in period $t$. We assume that $h_{st}^{i}$ is non-negative and, for each demand point $(i,t)$, is a nonincreasing function of $s$, i.e., holding inventory longer is always more costly. Thus, if the demand point $(i,t)$ is served by an order at time period $s$, then the system incurs a holding cost of $H_{st}^{i} := d_{it}h_{st}^{i}$.
\end{itemize}

The goal is to coordinate a sequence of (joint) orders to satisfy all the demand points on time so as to minimize the sum of joint ordering and holding costs over the $T$ periods. 


The above unified problem statement encompasses two classical deterministic inventory problems described below.

{\bf The submodular JRP.} The JRP involves multiple item types and a single retailer who faces demands. In each time step, any subset of item-types can be ordered incurring a joint ordering cost which is submodular. The objective is find a sequence of orders that satisfies all demands and minimizes the total  ordering and holding costs. The elements in the above problem statement are the item types in the JRP; and the joint ordering cost $f(\cdot)$ is commonly referred to as the setup cost (or equivalently, the fixed ordering cost) in the JRP.  

The submodular JRP considers a special class of $f(\cdot)$ called submodular functions (see, e.g., \cite{CELS2014}). More precisely, we assume that the function $f(\cdot)$ is non-negative, monotone non-decreasing, and also submodular. The non-negativity and monotonicity assert that for every $S_{1} \subseteq S_{2}  \subseteq \mathcal{N}$, we have $0 \le f(S_{1}) \le f(S_{2})$. Submodularity requires that  for every set $S_{1},S_{2} \subseteq \mathcal{N}$, we have 
$$
f(S_{1})+f(S_{2}) \ge f(S_{1} \cup S_{2})+f(S_{1} \cap S_{2}).
$$
There is an equivalent definition that conveys the economies of scale more clearly. That is, for every set $S_{1} \subseteq S_{2} \subseteq \mathcal{N}$ and any item type $i \in \mathcal{N}$, we have $f(S_{2}\cup \{i\}) - f(S_{2}) \le f(S_{1}\cup \{i\}) - f(S_{1})$, i.e., the additional cost of adding an item type to the joint order is decreasing as more item types have been included in that order. 

{\bf The IRP with arbitrary embedding metric.} The IRP  involves a single depot $r$ that stocks items, and a set of retailer locations (denoted by the ground set $\mathcal{N}$) facing demands. In each time step, any subset of locations can be visited using a vehicle originating from the depot. The objective here is to satisfy all demands while minimizing the sum of routing and holding costs. The elements in the above unified problem statement are the retailers in the IRP; and the joint ordering cost $f(\cdot)$ is the shipping or routing cost. 

The IRP is specified by a complete graph on vertices $V$ with a metric distance function $w: \binom{V}{2} \rightarrow \R_+$ that satisfies symmetry (i.e. $w(ba) =w(ab)$ for any $a,b\in V$) and triangle inequality (i.e. $w(ab)+w(bc) \ge w(ac)$ for any $a,b,c \in V$). The vertex set $V=\mathcal{N}\cup r$, containing the depot and the set of retailers. The shipping or routing cost $f(S)$ can be defined as the travelling salesman (TSP) cost of visiting the retailers in $S \subseteq V$. Formally,
\begin{equation}
\label{def:tsp}
f(S)\quad := \quad \mbox{minimum length of tour that visits each vertex in }S\cup\{r\},\qquad \forall S\subseteq \mathcal{N}.
\end{equation}

\subsection{IP Formulation and its LP Relaxation}
The unified problem described above can be written as an integer programming problem as follows (see also \cite{CELS2014}). First we define two types of binary variables $y_{s}^{S}$ and $x_{st}^{i}$ such that 
\begin{eqnarray*}
y_{s}^{S} &=&
\begin{cases}
1, & \text{if the subset of elements } S \subseteq \mathcal{N} \text{ is ordered in period } s,\\
0, & \text{otherwise} .
\end{cases} \\
x_{st}^{i} &=&
\begin{cases}
1, & \text{if the demand point } (i,t) \text{ is satisfied using an order from  period } s,\\
0, & \text{otherwise} .
\end{cases}
\end{eqnarray*}
Then the integer programming (IP) formulation is given by
\begin{eqnarray}
\label{IP}
\text{\bf (IP)} \qquad  \text{min} && \sum_{S\subseteq \mathcal{N}} \sum_{s=1}^{T} f(S)y_{s}^{S} + \sum_{(i,t)\in \mathcal{D}}  \sum_{s=1}^{t} H_{st}^{i}x_{st}^{i} \\
\text{s.t.} && \sum_{s=1}^{t} x_{st}^{i}=1, \qquad \forall  (i,t) \in \mathcal{D}\nonumber \\ 
&& x_{st}^{i} \le \sum_{S: i\in S \subseteq \mathcal{N} }y_{s}^{S}, \qquad  \forall  (i,t) \in \mathcal{D}, \forall s =1,\ldots, t  \nonumber \\ 
&& x_{st}^{i}, y_{s}^{S} \in \{0,1\}, \qquad  \forall  (i,t) \in \mathcal{D}, \forall s =1,\ldots, t, \forall  S \subseteq \mathcal{N}. \nonumber
\end{eqnarray}
The first constraint in (\ref{IP}) enforces that every demand point $(i,t)$ must be served by an order before or in time period $t$. The second constraint in (\ref{IP}) ensures that the joint order $S$ has to contain element $i$ if any demand $(i,t)$ is served at time period $s$. 
There is a natural linear programming (LP) relaxation of (IP) that relaxes the integer constraints on $x_{st}^{i}$ and $y_{s}^{S}$ to non-negativity constraints. 

To obtain approximation algorithms for the IP  (\ref{IP}) using our framework, we only need to assume that the set function $f(\cdot)$ satisfies an approximate notion of fractional subaddivitity (which is much weaker than submodularity).
\begin{definition}[$\beta$-approximate fractional subadditivity] 
\label{def_afs}
The set function $f(\cdot)$ is $\beta$-approximate fractional subadditive, if for any collection $\{S_i, \lambda_i\}$ of weighted subsets with $0\le \lambda_{i} \le 1$ and $\sum_{i|v \in S_{i}} \lambda_{i}\ge 1$ for each $v\in S$, we have $f(S) \le \beta \cdot \sum \lambda_{i} f(S_{i})$. Namely, if the sets $S_{i}$ form a \emph{fractional cover} of $S\sse {\cal N}$, then the cost of $S$ is at most $\beta$ times the sum of the costs $f(S_i)$ weighted by the corresponding coefficients. 
\end{definition}

It is known that if a function is submodular, then it is also fractional subadditive (see \cite{Feige2006}), i.e., the notion of submodularity is stronger. For the submodular JRP, the setup cost function $f(S)$ is submodular and hence also fractional subadditive (or equivalently, $1$-approximate fractional subadditive). 

For the IRP with arbitrary embedding metric, the vehicle routing cost $f(S)$ although not submodular, can be shown to be $1.5$-approximate fractional subadditive. This follows from the fact that the natural LP-relaxation for TSP has an integrality gap at most $1.5$ (see \cite{Wolsey1980} and \cite{Shmoys1990281}). 

Note that the LP relaxation of (\ref{IP}) has an exponential number of variables; we need to ensure that this LP relaxation can be (at least approximately) solved efficiently.
\begin{definition}[$\gamma$-approximate LP solution]
\label{def_alp}
We say that a feasible LP solution is $\gamma$-approximate, if its objective value is at most $\gamma$ times the optimal LP objective value.
\end{definition}
Using the ellipsoid method one can compute efficiently: an exact LP solution for the submodular JRP  and a $(2+ o(1))$-approximate LP solution for IRP. We delegate the discussion on this to Section \ref{sec_solve} for better readability of this paper.

\subsection{Our Main Results}

\begin{assumption}[$\alpha$-degree polynomial holding cost]
\label{assump1}
For each element $i \in \mathcal{N}$ and $1\le s \le t \in \mathcal{T}$, the holding cost of holding an inventory unit of element $i$ from period $s$ to $t$ is
$$
h_{st}^{i} =(t-s)^{\alpha}\bar{h}^{i}_t,
$$
for some base per-unit holding cost $\bar{h}^{i}_t > 0$ and some $\alpha \ge 1$. 
\end{assumption}
Note that when $\alpha=1$, this reduces to the conventional linear holding cost. We also have $H_{st}^{i} =d_{t}^{i}(t-s)^{\alpha}\bar{h}^{i}_t$.

Now we are in a position to formally state our main result (which will be proved in the following section).
\begin{theorem}
\label{mainresult}
Under Assumption \ref{assump1},  there is an $\mathcal{O}\left(\alpha \beta \gamma \cdot \frac{\log T}{\log \log T}\right)$-approximation algorithm for the integer program defined in (\ref{IP}), provided that $f(\cdot)$ is $\beta$-approximate fractional subadditive, and a $\gamma$-approximate solution to the LP relaxation of \eqref{IP} can be found in polynomial time. 
\end{theorem}

Corollary \ref{mainresult2} below is an immediate consequence of Theorem \ref{mainresult}, since 
\begin{enumerate}
\item $\beta=\gamma=1$ for the submodular JRP;
\item $\beta=1.5$ and $\gamma=2+o(1)$ for the IRP.
\end{enumerate}

\begin{corollary}
\label{mainresult2}
Under Assumption \ref{assump1},  there is an $\mathcal{O}\left(\frac{\log T}{\log \log T}\right)$-approximation algorithm for both the submodular JRP and the IRP with arbitrary embedding metric.
\end{corollary}

To the best of our knowledge, this is the first sub-logarithmic approximation ratio for either problem.

We remark that it is immediate that our approach yields $\mathcal{O}(\log T)$-approximation algorithms for submodular JRP and IRP with arbitrary (monotone) holding costs (i.e., waiving Assumption \ref{assump1}). 


\section{LP-Rounding Algorithm} \label{sec_algorithm}
We present an LP-rounding algorithm for the integer program (\ref{IP}) under Assumption \ref{assump1} in Section \ref{subsec_algo}, and then carry out a worst-case performance analysis in Section \ref{subsec_analysis}.

\subsection{Algorithm Description} \label{subsec_algo}
We describe our procedure of rounding a $\gamma$-approximate solution $(y,x)$ of (LP). We set $\rho := \lfloor (\log T)^{1/(2\alpha)} \rfloor $.

\paragraph{\bf Step 1 -- Constructing extended  shadow intervals.}
We first construct what-we-call \emph{extended shadow intervals} as follows. For each demand point $(i,t)$, we take the $\gamma$-approximate LP solution and find a time period $s^{\prime}_{(i,t)}$ such that 
$$
\sum_{s=s^{\prime}_{(i,t)}}^{t}x_{st}^{i} \ge 1/2 \text{ and } \sum_{s=s^{\prime}_{(i,t)}+1}^{t}x_{st}^{i} < 1/2,
$$
i.e., finding the closest $s$ to the left of $t$ such that the sum of $x$ variables for $(i,t)$ contains a half point. Then $[s^{\prime}_{(i,t)},t]$ is called the \emph{shadow interval} for this particular demand point $(i,t)$. 
We also measure its length $l_{(i,t)}: = t-s^{\prime}_{(i,t)}$.

Next for each demand point $(i,t)$, we round the length $l_{(i,t)}$ up to the nearest power of $\rho$. If $s^{\prime}_{(i,t)}=t$ then we set $s^{*}_{(i,t)}=t$. Else we find the smallest integer $m\ge 1$ such that
$\rho^{m} \ge t-s^{\prime}_{(i,t)},$ and then stretch the original shadow interval from $t$ to $s^{*}_{(i,t)}$ where
$$
s^{*}_{(i,t)} = t - \rho^{m}  \le s^{\prime}_{(i,t)}.
$$ 
We call the interval $[s^{*}_{(i,t)},t]$ the \emph{extended shadow interval} for the demand point $(i,t)$, and also measure its length $l^{*}_{(i,t)} := t-s^{*}_{(i,t)}$. Figure \ref{fig_si} below gives a graphical representation of this step.

\begin{figure}[htbp]
\begin{center}
\includegraphics[width=12cm]{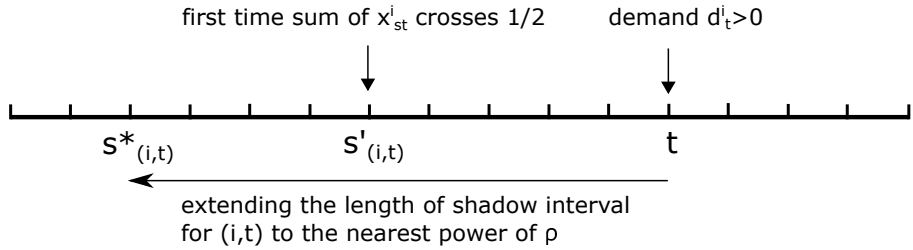}
\caption{Illustration of an extended shadow interval for demand point $(i,t)$.}
\label{fig_si}
\end{center}
\end{figure}

\paragraph{\bf Step 2 -- Partitioning demand points according to extended shadow intervals.}
Next we partition the demand points according to the length of their extended shadow intervals. For each demand point $(i,t)$, its length $l^{*}_{(i,t)}$ falls into exactly one of the values below (recall by construction $l^{*}_{(i,t)}$ is either zero or an integer power of $\rho$).
$$
\{0, \rho^{1}, \rho^{2}, \ldots ,  \rho^{k-2} , \rho^{k-1} \wedge T \}, \text{ where } k= 1+\left\lceil \log_{\rho} T \right\rceil.
$$
In this way, we have partitioned the demand points into $k = \mathcal{O}\left(\alpha \frac{\log T}{\log \log T}\right)$ number of groups as follows,
\begin{eqnarray*}
\mathcal{L}_{0} = \left\{(i,t) \in \mathcal{D}: l^*_{(i,t)} = 0 \right\}  \quad \mbox{and}\quad 
\mathcal{L}_{m} = \left\{(i,t) \in \mathcal{D}: l^*_{(i,t)} = \rho^{m} \wedge T) \right\}, \,\, \forall \; m \in \{1,\ldots,k-1\}.
\end{eqnarray*} 
That is, the shadow intervals within each group $\mathcal{L}_{m}$ share the same length: 
$$w_m := \left\{ 
\begin{array}{ll}
0 & \mbox{ if } m=0\\
\rho^m \wedge T & \mbox{ if } m \in \{1,2,\ldots,k-1\}
\end{array}
\right.$$

\paragraph{\bf Step 3 -- Placing orders.}  Based on the above partition of demand points, we describe our ordering procedure.  

Now fix an $m \in \{0,1,\ldots,k-1\}$ and focus on the demand group $\mathcal{L}_{m}$. Let $\tau_{j} = 1 + j\cdot w_m$ $(\le T)$ for $j=0,1,\ldots$.  We place a tentative (joint) order in each period $\tau_{j}$ $(j=0,1, \ldots)$, i.e., once every $w_m$ periods. 
In each period $\tau_{j} \le T$ $(j=0,1,2  \ldots)$, we identify the set of elements
$$
A_{m}^{j} = \left\{i: (i,t)\in \mathcal{L}_{m} \text{ and } \tau_{j} \in [s^{*}_{(i,t)},t]   \right\},
$$
i.e., all the elements within $\mathcal{L}_{m}$ whose shadow intervals contain (or intersect with) time period $\tau_{j}$. We then place an actual joint order that serves the demand points associated with $A_{m}^{j}$ in period $\tau_{j}$. Figure \ref{fig2} gives one specific example of how the algorithm places these orders.

\ignore{
Then starting from period $\tau_{0}=1$, we identify  the set of elements (i.e., item types in the JRP and retailers in the IRP)
$$
A_{m}^{0} = \left\{i: (i,t)\in \mathcal{L}_{m} \text{ and } 1\in [s^{*}_{(i,t)},t] \right\},
$$
i.e., all the elements within $\mathcal{L}_{m}$ whose shadow intervals contain (or intersect with) time period $1$. We then place an actual joint order that serves the demand points associated with $A_{m}^{0}$ in period $1$.

In each subsequent period $\tau_{j} \le T$ $(j=1,2, \ldots)$, we identify the set of elements
$$
A_{m}^{j} = \left\{i: (i,t)\in \mathcal{L}_{m} \text{ and } \tau_{j} \in [s^{*}_{(i,t)},t]  \text{ and }  (i,t) \notin A_{m}^{0}\cup \ldots \cup A_{m}^{j-1}\right \},
$$
i.e., all the elements within $\mathcal{L}_{m}$ whose shadow intervals contain (or intersect with) time period $\tau_{j}$ but have not been ordered yet. We then place an actual joint order that serves the demand points associated with $A_{m}^{j}$ in period $\tau_{j}$. Figure \ref{fig2} gives one specific example of how the algorithm places these orders.
}

We repeat the above procedure for all groups $m=0,1,  \ldots, k-1$. In any given period, if there is more than one joint order (across different groups), we simply merge them into a single joint order.

\begin{figure}[htbp]
\begin{center}
\includegraphics[scale=0.65]{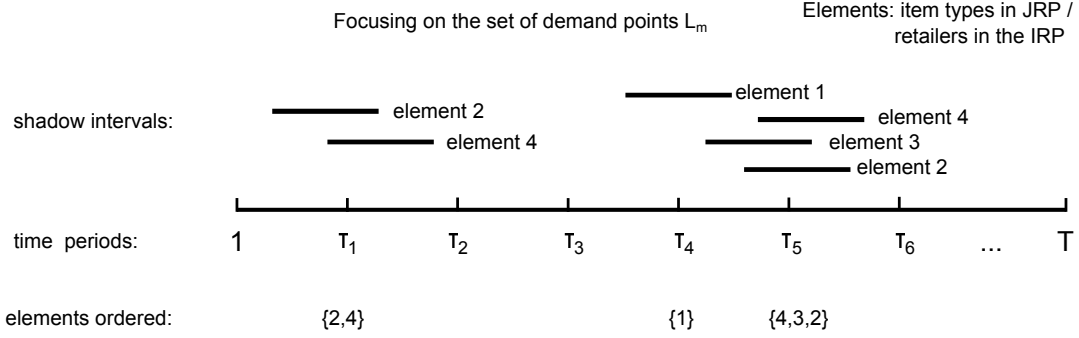}
\caption{Placing actual orders for the demand points within $\mathcal{L}_{m}$}
\label{fig2}
\end{center}
\end{figure}

This concludes the description of our LP-rounding algorithm.

\subsection{Worst-case Analysis} \label{subsec_analysis} 
We shall prove that our LP-rounding algorithm gives an $\mathcal{O}\left(\alpha \beta \gamma \frac{\log T}{\log \log T}\right)$-approximation for the unified problem. For brevity, we just call the $\gamma$-approximate solution $(y,x)$ to the LP relaxation of (\ref{IP}) \emph{the $\gamma$-approximate LP solution}.

\subsubsection*{Analysis of Holding Cost}

\begin{lemma} \label{lemma_holding_1}
Let $[s^{\prime}_{(i,t)},t]$ be the shadow interval for some demand point $(i,t)$. Then we have 
\begin{eqnarray}
H_{s^{\prime}_{(i,t)}t}^{i} \le  2\cdot \sum_{s=1}^{t} H_{st}^{i}x_{st}^{i},
\end{eqnarray}
where $x_{st}^{i}$'s are in the  $\gamma$-approximate LP solution.
\end{lemma}

\begin{proof}
By the construction of shadow intervals, we have $\sum_{s=s^{\prime}_{(i,t)}+1}^{t}x_{st}^{i} < 1/2$. Now since $\sum_{s=1}^{t}x_{st}^{i}=1$ by the first constraint in (\ref{IP}), we must have 
$\sum_{s=1}^{s^{\prime}_{(i,t)}} x_{st}^{i} \ge 1/2$. Hence we have
\begin{eqnarray}
H_{s^{\prime}_{(i,t)}t}^{i} \le 2 \cdot \sum_{s=1}^{s^{\prime}_{(i,t)}} H_{s^{\prime}_{(i,t)}t}^{i}x_{st}^{i} \le 2 \cdot \sum_{s=1}^{s^{\prime}_{(i,t)}} H_{st}^{i}x_{st}^{i} \le  2 \cdot \left(\sum_{s=1}^{s^{\prime}_{(i,t)}} H_{st}^{i}x_{st}^{i} + \sum_{s=s^{\prime}_{(i,t)}+1}^{t} H_{st}^{i}x_{st}^{i} \right) = 2\cdot \sum_{s=1}^{t} H_{st}^{i}x_{st}^{i},
\end{eqnarray}
where the second inequality is due to  $H_{s^{\prime}_{(i,t)}t} \le H_{st}$ for all $s\le s^{\prime}_{(i,t)}$ by monotonicity of holding costs.
\end{proof}

\begin{lemma} \label{lemma_holding_poly}
The total holding cost for the solution found by the LP-rounding algorithm is at most 
$$
\mathcal{O}( \sqrt{\log T}) \cdot \sum_{(i,t)\in \mathcal{D}}  \sum_{s=1}^{t} H_{st}^{i}x_{st}^{i}, 
$$
where $x_{st}^{i}$'s are in the  $\gamma$-approximate LP solution.
\end{lemma}

\begin{proof}
By the polynomial holding cost structure, for each demand point $(i,t)$ in some $\mathcal{L}_{m}$, we have 
\begin{eqnarray*}
H_{s^{*}_{(i,t)}t}^{i} &=& d_{t}^{i} (t-s^{*}_{(i,t)})^{\alpha} =  d_{t}^{i} (w_m)^{\alpha} \bar{h}^{i},  \\
H_{s^{\prime}_{(i,t)}t}^{i} &=& d_{t}^{i} (t- s^{\prime}_{(i,t)})^{\alpha} \ge  d_{t}^{i} (w_m/\rho)^{\alpha} \bar{h}^{i}, 
\end{eqnarray*}
where the inequality follows from the construction of extended shadow intervals. Hence it is clear that 
$$H_{s^{*}_{(i,t)}t}^{i} \le  \rho^{\alpha} H_{s^{\prime}_{(i,t)}t}^{i}.$$

By the LP-rounding algorithm, for each demand point $(i,t)$, we must have placed a (joint) order containing element $i$ inside its extended shadow interval $[s^{*}_{(i,t)},t]$. Due to monotonicity of holding costs, the worst-case (that gives the highest possible holding cost) happens when our algorithm places the order at exactly time period $s^{*}_{(i,t)}$ to satisfy the demand point $(i,t)$.  Hence, the total holding cost associated with the demand point $(i,t)$ is upper bounded by 
\begin{eqnarray*}
H_{s^{*}_{(i,t)}t}^{i} \le  \rho^{\alpha} H_{s^{\prime}_{(i,t)}t}^{i} \le 2\rho^{\alpha} \cdot \sum_{s=1}^{t} H_{st}^{i}x_{st}^{i},  
\end{eqnarray*}
where the second inequality follows from Lemma \ref{lemma_holding_1}. Now setting $\rho = \lfloor (\log T)^{1/(2\alpha)} \rfloor$ yields the result. 
\end{proof}

The intuition behind Lemma \ref{lemma_holding_poly} is that when we stretch the shadow interval to the nearest integer power of $\rho$, the holding cost within the extended shadow interval does not grow too large due to the polynomial holding cost structure. In particular it grows by at most a factor of $\mathcal{O}(\sqrt{\log T})$.  On the other hand, stretching the shadow intervals in this manner gives us a tighter bound on the ordering cost (as shown below).

\subsubsection*{Analysis of Ordering Cost}
To analyze the ordering cost component, we introduce the following bridging problem: 
\begin{eqnarray}
\label{Covering}
\text{\bf (Covering-LP)} \qquad \text{min} && \sum_{S\subseteq \mathcal{N}} \sum_{s=1}^{T} f(S)z_{s}^{S}  \\
\text{s.t.}  && \sum_{s=s^{*}_{(i,t)}}^{t} \sum_{S: i\in S \subseteq \mathcal{N} } z_{s}^{S} \ge 1, \qquad  \forall  (i,t) \in \mathcal{D}  \nonumber \\ 
&& z_{s}^{S} \ge 0, \qquad  \forall s =1,\ldots, t, \forall  S \subseteq \mathcal{N}. \nonumber
\end{eqnarray}
The intuition behind introducing this bridging problem is as follows: if our algorithm places an order to satisfy the demand within its extended shadow interval then Lemma \ref{lemma_holding_poly} implies that the holding cost can be bounded by $\mathcal{O}(\sqrt{\log T})$ times the LP holding cost. Thus, the problem reduces to finding a ``cover'' for these intervals as defined in Problem (\ref{Covering}). In the remainder of the worst-case analysis, we will focus on analyzing this Covering-LP.

\begin{lemma} \label{lemma_clp1}
The optimal objective value of the Covering-LP is at most 
$$
2\cdot \sum_{S\subseteq \mathcal{N}} \sum_{s=1}^{T} f(S)y_{s}^{S}.
$$
where $y_{s}^{S}$'s are in the  $\gamma$-approximate LP solution.
\end{lemma}

\begin{proof}
We first check $\bar{z}_{s}^{S} = 2y_{s}^{S}$ (where $y_{s}^{S}$ is the  $\gamma$-approximate LP solution) is feasible to the Covering-LP defined in (\ref{Covering}). It is obvious that $\bar{z}_{s}^{S} = 2y_{s}^{S} \ge 0$ for all $s =1,\ldots, t$ and for all $S \subseteq \mathcal{N}$. It suffices to verify the first set of constraints. Indeed, for each $(i,t) \in \mathcal{D}$, we have
\begin{equation}
\sum_{s=s^{*}_{(i,t)}}^{t} \sum_{S: i\in S \subseteq \mathcal{N} } \bar{z}_{s}^{S}=\sum_{s=s^{*}_{(i,t)}}^{t} \sum_{S: i\in S \subseteq \mathcal{N} } 2y_{s}^{S} \ge 
\sum_{s=s^{*}_{(i,t)}}^{t} 2x_{st}^{i} \ge 1,
\end{equation}
where the first inequality follows from the second constraint in (\ref{IP}), and the second inequality follows from the fact that $\sum_{s=s^{*}_{(i,t)}}^{t} x_{st}^{i} \ge \sum_{s=s^{\prime}_{(i,t)}}^{t} x_{st}^{i} \ge 1/2$ (by the definition of shadow intervals and their extensions).

Hence, the optimal objective value of the Covering-LP 
\begin{equation}
\sum_{S\subseteq \mathcal{N}} \sum_{s=1}^{T} f(S)z_{s}^{S} \le \sum_{S\subseteq \mathcal{N}} \sum_{s=1}^{T} f(S)\bar{z}_{s}^{S} = 2\sum_{S\subseteq \mathcal{N}} \sum_{s=1}^{T} f(S)y_{s}^{S},
\end{equation}
where the first inequality follows from that $z_{s}^{S}$ is optimal while $\bar{z}_{s}^{S}$ is feasible.
\end{proof}

Fix an $m\in \{0,1,\ldots, k-1\}$ and we focus our attention on the demand group $\mathcal{L}_{m}$ (which have equal length of shadow intervals). We shall show that the total ordering cost associated with the set $\mathcal{L}_{m}$ by our LP-rounding algorithm can be upper bounded by $2\beta$ times the Covering-LP cost. Our proof strategy relies on the notion of approximate fractional subadditivity (see Definition \ref{def_afs}).

\begin{lemma} \label{lemma_clp2}
The total ordering cost associated with the set $\mathcal{L}_{m}$ by our LP-rounding algorithm is at most 
$$
2 \beta \cdot \sum_{S\subseteq \mathcal{N}} \sum_{s=1}^{T} f(S)z_{s}^{S}.
$$
where $z_{s}^{S}$'s are the optimal Covering-LP solutions.
\end{lemma}

\begin{proof}
Recall that for each demand group $\mathcal{L}_{m}$, the LP-rounding algorithm places a tentative (joint) order in each period $\tau_{j} \le T$ $(j=0,1,\ldots)$. Then in each time period $\tau_{j}$ the algorithm identifies the elements $A^j_m$ within $\mathcal{L}_{m}$ whose shadow intervals contain (or intersect with) $\tau_{j}$ 
and places an actual (joint) order $A_{m}^{j}$ that includes all of these ``intersecting" elements. 

Now take any $\tau_{j} \le T$: since the length of the shadow intervals in $\mathcal{L}_{m}$ is exactly $w_m$, all the shadow intervals associated with the order $A_{m}^{j}$ must lie within the interval $(\tau_{j-1}, \tau_{j+1}]$ (see Figure \ref{fig2} as an example). Our LP-rounding algorithm places an actual (joint) order $A_{m}^{j}$ in period $\tau_{j}$ and incurs an ordering cost $f(A_{m}^{j})$. We will show that  the Covering-LP provides us with a fractional cover of $A_{m}^{j}$ which will be used to upper bound $f(A_{m}^{j})$.

Indeed, for each demand point $(i,t)$ associated with the order $A_{m}^{j}$, we have
\begin{eqnarray}
\sum_{S:i\in S \subseteq \mathcal{N}} \sum_{s>\tau_{j-1}}^{\tau_{j+1}} z_{s}^{S} =  \sum_{s>\tau_{j-1}}^{\tau_{j+1}}  \sum_{S:i\in S \subseteq \mathcal{N}} z_{s}^{S} 
\ge \sum_{s=s^{\prime}_{(i,t)}}^{t} \sum_{S:i\in S \subseteq \mathcal{N}} z_{s}^{S} \ge 1,
\end{eqnarray}
where the first inequality holds because every shadow interval associated with the order $A_{m}^{j}$ must lie within the interval $(\tau_{j-1}, \tau_{j+1}]$ and the last inequality follows from the first constraint in the Covering-LP (\ref{Covering}).

Since $f(\cdot)$ is $\beta$-approximate fractional subadditive, then according to Definition \ref{def_afs}, 
\begin{eqnarray}
f(A_{m}^{j}) \le \beta \cdot \sum_{S: S \subseteq \mathcal{N}} \sum_{s>\tau_{j-1}}^{\tau_{j+1}} z_{s}^{S} f(S).
\end{eqnarray}
It is then immediate that the ordering cost associated with the set $\mathcal{L}_{m}$ by our LP-rounding algorithm 
\begin{eqnarray}
\sum_{j\ge 0}f(A_{m}^{j}) \,\, \le  \,\, \beta \cdot \sum_{S: S \subseteq \mathcal{N}} \sum_{j\ge 0} \sum_{s>\tau_{j-1}}^{\tau_{j+1}} z_{s}^{S} f(S)  \,\, \le  \,\,  2 \beta \cdot \sum_{S: S \subseteq \mathcal{N}} \sum_{s=1}^{T} z_{s}^{S} f(S).
\end{eqnarray}
\end{proof}

\begin{lemma} \label{lemma_clp3}
The total ordering cost for the solution by our LP-rounding algorithm is at most 
$$
\mathcal{O}\left(\alpha \beta \frac{\log T}{\log \log T}\right) \cdot \sum_{S\subseteq \mathcal{N}} \sum_{s=1}^{T} f(S)y_{s}^{S},
$$
where $y_{s}^{S}$'s are in the  $\gamma$-approximate LP solution.
\end{lemma}

\begin{proof}
By Lemmas \ref{lemma_clp1} and \ref{lemma_clp2}, for each group $\mathcal{L}_{m}$ $(m=0,1, \ldots,k-1)$, we conclude that the total ordering cost associated with the set $\mathcal{L}_{m}$ in our LP-rounding algorithm is at most 
$$
2\beta \cdot \sum_{S\subseteq \mathcal{N}} \sum_{s=1}^{T} f(S)z_{s}^{S} \le 4 \beta \cdot \sum_{S\subseteq \mathcal{N}} \sum_{s=1}^{T} f(S)y_{s}^{S},
$$ 
where $z_{s}^{S}$'s are the optimal Covering-LP solution and $y_{s}^{S}$'s are the  $\gamma$-approximate LP solution. Then the result follows from the fact that the number of groups $k = \mathcal{O}\left( \alpha \frac{\log T}{\log \log T}\right)$.
\end{proof}

Now we are ready to prove our main result Theorem \ref{mainresult}.

\noindent{\bf Proof of Theorem \ref{mainresult}:}
Combining the results from Lemmas \ref{lemma_holding_poly} and \ref{lemma_clp3}, the total holding and ordering costs for the solution by our LP-rounding algorithm is at most
$\mathcal{O}\left(\alpha \beta \frac{\log T}{\log \log T}\right)$ times the  $\gamma$-approximate LP solution. 
\hfill$\blacksquare$

\noindent 
{\bf Remark:} The $\mathcal{O}\left(\frac{\log T}{\log \log T}\right)$ approximation ratio is the best tradeoff achievable (in our approach) between the loss in holding and ordering costs, even under linear holding costs. Recall that for a given set $W$ of widths for extended shadow intervals, the loss in ordering cost is just the number $|W|$ of distinct widths and the loss in holding cost depends on the aggregate stretch-factor incurred when the width of each shadow interval is increased to a value in $W$. Even if we allow for an arbitrary set $W$ of widths (that may depend on the LP solution) and compute the worst ratio (using a ``factor revealing linear program'' as in~\cite{JMMSV03}) then we obtain $\mathcal{O}\left(\frac{\log T}{\log \log T}\right)$ as the approximation ratio.

\subsubsection*{A Special Case with Perishable Goods} We now consider a special holding cost which models perishable items with a fixed life-time $c>0$. For each demand point $(i,t)$, we can only start satisfying this order $c$ periods before $t$, i.e., the ordering window is $[t-c, t]$. This setting is equivalent to the following holding cost structure. For each $i \in \mathcal{N}$ and $1\le s \le t \in \mathcal{T}$, 
\begin{equation*}
h_{st}^{i} =
\begin{cases}
0 & \text{if } t-c \le s \le t ,\\
\infty & \text{if } s < t-c.
\end{cases}
\end{equation*}
We also have $H_{st}^{i} =d_{t}^{i}h_{st}^{i}$.

In this setting, for each demand point $(i,t)$, the extended shadow interval is simply $[t-c, t]$ with length $c$. Hence our LP-rounding algorithm and its worst-analysis will apply with just a single group, and we obtain:
\begin{theorem} \label{thm_perishable}
When items are perishable with a fixed life-time and the holding cost is negligible, the LP-rounding algorithm gives a $2$-approximation for the submodular JRP, and a $(6+o(1))$-approximation for the IRP.
\end{theorem}

\section{Solving the LP Relaxation} \label{sec_solve}
As mentioned earlier in Section \ref{sec_model}, the LP relaxation has an exponential number of variables and we first argue that there is an efficient way of solving this LP. We can readily write the dual of (LP) as 
\begin{eqnarray}
\label{Dual-LP}
\text{\bf (DLP)} \qquad \text{max} && \sum_{(i,t) \in \mathcal{D}} b_{t}^{i} \\
\text{s.t.} && b_{t}^{i} \le H_{st}^{i} + \bar{b}_{st}^{i}, \qquad \forall  (i,t) \in \mathcal{D}, \forall s =1,\ldots, t   \nonumber \\ 
&& f(S) - \sum_{i \in S}\sum_{t=s}^{T} \bar{b}_{st}^{i} \ge 0, \qquad  \forall s =1,\ldots, T, \forall S \subseteq \{1,\ldots N\}  \nonumber \\ 
&& \bar{b}_{st}^{i} \ge 0, \qquad  \forall  (i,t) \in \mathcal{D}, \forall s =1,\ldots, t. \nonumber
\end{eqnarray}
The $b_{t}^{i}$ and $\bar{b}_{st}^{i}$ are the dual variable corresponding to the first and second constraints in the LP relaxation of (\ref{IP}). Note that the dual formulation (\ref{Dual-LP}) has an exponential number of constraints. 

\subsubsection*{Submodular JRP}
In the submodular JRP, the left hand side of the second constraint $f(S) - \sum_{i \in S}\sum_{t=s}^{T} \bar{b}_{st}^{i}$ is clearly submodular. Thus, there is an efficient separation oracle by using submodular function minimization~\cite{S-book} to find violated constraints. This implies that the dual problem (and therefore the primal) can be efficiently solved using the ellipsoid method. This was also discussed in \cite{CELS2014}. 

\subsubsection*{Approximately Solving the LP for IRP}
The TSP costs $f(\cdot)$ are not submodular, and in fact the above separation problem is NP-hard. However, there is an approximate separation oracle (see Lemma \ref{oracle}) which suffices to compute an approximately optimal solution to (DLP).  The main ingredient is an approximation algorithm for the following auxiliary problem:
\begin{definition}[Minimum ratio TSP]
The input is a metric $(V,w)$ with a designated depot $r\in V$ and rewards $a:V\rightarrow \mathbb{R}_+$ on vertices. The goal is to find a subset $S\sse V$ that minimizes: 
$$\frac{f(S)}{a(S)}, \qquad \mbox{where $f(S)$ is the TSP cost as defined in~\eqref{def:tsp} and } a(S)=\sum_{i\in S} a_i.$$
\end{definition}

\begin{theorem}[Garg \cite{Garg2005}]\label{thm:rTSP}
There is a $(2+o(1))$-approximation algorithm for the minimum ratio TSP problem.
\end{theorem}
\begin{proof}
The algorithm for minimum ratio TSP uses the $2$-approximation algorithm for the related $k$-TSP problem (i.e. given a metric, depot $r$ and target $k$, find a minimum length tour from the depot that visits at least $k$ vertices). The algorithm which is based on standard scaling arguments, is given below for completeness:
\begin{enumerate}
\item Guess (by enumerating over $|V|$ choices) the maximum reward vertex $u$ in an optimal solution.
\item Remove all vertices with reward more than $a_u$.
\item For each $v\in V$ set its new reward $\bar{a}_v$ to be the largest integer such that $\bar{a}_v\cdot \frac{a_u}{n^2}\le a_v$.
\item For each $k=1,\cdots, n^3$, run the $k$-TSP algorithm with target $k$ on the modified metric containing $\bar{a}_v$ co-located vertices at each $v\in V$.
\item Output the best ratio solution found (over all choices of $u$ and $k$).
\end{enumerate}
It is easy to see that this algorithm runs in polynomial time since each $\bar{a}_v \le n^2$. We now show that it has an approximation ratio of $\gamma=2+o(1)$. Let $S^*$ denote an optimal solution and $u\in S^*$ the maximum reward vertex. Consider the run of the above algorithm for this choice of $u$: note that none of the vertices from $S^*$ is removed. By the definition of new rewards, we have $a_v -\frac{a_u}{n^2} <  \bar{a}_v\cdot \frac{a_u}{n^2}\le a_v$ for all $v\in S^*$. So $\frac{a_u}{n^2}\sum_{v\in S^*} \bar{a}_v > a(S^*) - \frac{a_u}{n}\ge (1-1/n)a(S^*)$, which implies (as $\bar{a}$ is integer valued) that  $\bar{a}(S^*)\ge k := \lceil \frac{n^2}{a_u}(1-\frac1n)a(S^*)\rceil$. For this choice of $k$, the $k$-TSP algorithm is guaranteed to find a subset $S\sse V$ with $\bar{a}(S)\ge k$ and $f(S)\le 2\cdot f(S^*)$. The ratio of this solution is:
$$\frac{f(S)}{a(S)} \le \frac{2f(S^*)}{a(S)} \le \frac{n^2}{a_u}\cdot \frac{2f(S)}{\bar{a}(S)} \le \frac{2}{1-1/n}\cdot \frac{f(S^*)}{a(S^*)}.$$
Hence the above algorithm achieves a $2+o(1)$ approximation guarantee for minimum ratio TSP. 
\end{proof}

\def\bb{\mathbf{b}}
\def\bbar{\mathbf{\bar{b}}} 

Let $\gamma=2+o(1)$ denote the approximation guarantee for minimum ratio TSP from Theorem~\ref{thm:rTSP}.
We now show that this algorithm leads to an approximate separation oracle.

\begin{lemma}
\label{oracle}
There is a polynomial time algorithm, that given vectors $\bb= \{b^i_t : (i,t)\in \cD\}$ and $\bbar=\{\bar{b}^i_{st} : (i,t)\in\cD, 1\le s\le t\}$, outputs one of the following:
\begin{enumerate}
\item A constraint in (DLP) that is violated by $(\bb,\bbar)$.
\item Certificate that $(\frac1\gamma \bb,\, \frac1\gamma \bbar)$ is feasible in (DLP).
\end{enumerate}
\end{lemma}
\begin{proof}
The first set of constraints in (DLP) and non-negativity of $\bar{b}^i_{st}$ are easy to verify since they are polynomial in number. 
Below we assume that these are satisfied.

In order to verify the second set of constraints in (DLP), we use Theorem~\ref{thm:rTSP}. For each $s\in [T]$ define an instance $\cI_s$ of minimum ratio TSP on metric $(V,w)$, depot $r$ and with rewards $a_v:=\sum_{t=s}^T \bar{b}^v_{st}$ for all $v\in V$. Let $A_s$ denote the solution found by the approximation algorithm of Theorem~\ref{thm:rTSP} on $\cI_s$.

If any solution $A_s$ has ratio less than one then it provides a violated constraint for $(\bb,\bbar)$. This corresponds to the first condition in the lemma.

If every solution $A_s$ has ratio at most one, we will show that  the second condition in the lemma holds. The non-negativity constraints are clearly satisfied by $(\frac1\gamma \bb,\, \frac1\gamma \bbar)$. To check the first set of constraints in (DLP), note that for any $(i,t)\in\cD$ and $s\in[T]$,
$$\frac{1}{\gamma}\left( b^i_t-\bar{b}^i_{st} \right) \le \max\left\{ 0,\, b^i_t-\bar{b}^i_{st} \right\} \le H^i_{st}.$$

To check the second set of constraints, note that for any $s\in [T]$ we have by the approximation guarantee in Theorem~\ref{thm:rTSP},
$$\min_{S\sse V} \, \frac{f(S)}{\sum_{v\in S} \sum_{t=s}^T \bar{b}^v_{st}} \ge \frac{1}{\gamma}.$$
This implies that $(\frac1\gamma \bb,\, \frac1\gamma \bbar)$ satisfies all these constraints.  
\end{proof}

Using the above separation oracle for (DLP) within the ellipsoid algorithm, we obtain a  $\gamma$-approximately optimal solution to (DLP), see eg.~\cite{J03}.  Then solving (LP) restricted to the (polynomially many) variables that are dual to the constraints generated in solving (DLP), we obtain a $\gamma$-approximately optimal solution to (LP) as well.

\medskip
\paragraph{Running time.}
Using the linear programming algorithms in~\cite{V90,V96} along with some preprocessing, the running time of the above  approach is dominated by $\tilde{O}(\cD^3 T^3)$ plus $\tilde{O}(\cD T^2)$ calls\footnote{The $\tilde O$ notation hides logarithmic factors.} to a subroutine for:
\begin{itemize}
\item submodular function minimization in case of JRP. 
\item minimum-ratio TSP in case of IRP.
\end{itemize}


\section{Concluding Remarks} 
\label{sec_cf}
We presented an $\mathcal{O}\left(\frac{\log T}{\log \log T}\right)$-approximation algorithm for submodular-JRP and IRP when holding costs are polynomial functions. Moreover, this approach applies to any ordering cost for which the correponding LP relaxation can be solved approximately and the ordering cost satisfies an approximate notion of fractional subadditivity.  Obtaining a constant-factor approximation algorithm for submodular JRP and IRP on general metrics (even with linear holding costs) remain the main open questions.

\section*{Acknowledgments}
The authors have benefited from valuable comments from and discussions on this work with Retsef Levi (MIT). 

\bibliographystyle{ormsv080}
\bibliography{JRPsub} 

\end{document}